\documentclass{article}
\usepackage{amsmath,amssymb,amsthm,verbatim}
\usepackage[margin=1in]{geometry}

\title{A Polylogarithmic PRG for Degree $2$ Threshold Functions in the Gaussian Setting}
\author{Daniel M. Kane\\
Department of Mathematics\\
Stanford University\\
dankane@math.stanford.edu}

\newcommand{\R}{\mathbb{R}}

\newcommand{\pr}{\textrm{Pr}}
\newcommand{\sgn}{\textrm{sgn}}
\newcommand{\E}{\mathbb{E}}
\newcommand{\Z}{\mathbb{Z}}

\newcommand{\erf}{\textrm{erf}}

\newtheorem{thm}{Theorem}
\newtheorem{prop}[thm]{Proposition}
\newtheorem{cor}[thm]{Corollary}
\newtheorem{lem}[thm]{Lemma}

\newtheorem*{defn}{Definition}

\begin{document}
\maketitle

\section{Introduction}We say that a function $f:\R^n\rightarrow\{+1,-1\}$ is a (degree-$d$) \emph{polynomial threshold function} (PTF) if it is of the form $f(x)=\sgn(p(x))$ for $p$ some (degree-$d$) polynomial in $n$ variables.  Polynomial threshold functions make up a natural class of Boolean functions and have applications to a number of fields of computer science such as circuit complexity \cite{curciutApp}, communication complexity \cite{commApp} and learning theory \cite{learningApp}.

In this paper, we study the question of pseudorandom generators (PRGs) for polynomial threshold functions of Gaussians (and in particular for $d=2$).  In other words, we wish to find explicit functions $F:\{0,1\}^s\rightarrow \R^n$ so that for any degree-$2$ polynomial threshold function $f$
$$
\left| \E_{x\sim_u \{0,1\}^s}[ f(F(x))] - \E_{X\sim \mathcal{G}^n}[f(X)]\right| < \epsilon.
$$
We say that such an $F$ is a pseudorandom generator of seed length $s$ that fools degree-$d$ polynomial threshold functions with respect to the Gaussian distribution to within $\epsilon$.  In this paper, we develop a generator with $s$ polylogarithmic in $n$ and $\epsilon$ in the case when $d=2$.

\subsection{Previous Work}

There have been a number of papers dealing with the question of finding pseudorandom generators for polynomial threshold functions with respect to the Gaussian distribution or the Bernoulli distribution (i.e. uniform over $\{-1,1\}^n$).  Several early works in this area showed that polynomial threshold functions of various degrees could be fooled by arbitrary $k$-wise independent families of Gaussian or Bernoulli random variables.  It should be noted that a $k$-wise independent family of Bernoulli random variables can be generated from a seed of length $O(k\log(n))$.  Although, any $k$-wise independent family of Gaussians will necessarily have infinite entropy, it is not hard to show that a simple discretization of these random variables leads to a generator of comparable seed length.  These results on fooling polynomial threshold functions with $k$-independence are summarized in Table \ref{kIndepTable} below.
\begin{table}[h]
\caption{Generators Based on Limited Independence}
\begin{tabular}{|l|c|c|l|}
\hline
Paper & Bernoulli/Gaussian & d & k \\
\hline
Diakonikolas, Gopalan, Jaiswal, Servedio, Viola \cite{IndepHalf} & Bernoulli & 1 & $O(\epsilon^{-2}\log^2(\epsilon^{-1}))$ \\
Diakonikolas, Kane, Nelson \cite{IndepDeg2} & Gaussian & 1 & $O(\epsilon^{-2})$ \\
Diakonikolas, Kane, Nelson \cite{IndepDeg2} & Both & 2 & $O(\epsilon^{-8})$\footnotemark\\
Kane \cite{kIndep} & Both & $d$ & $O_d\left( \epsilon^{-2^{O(d)}}\right)$ \\
\hline
\end{tabular}\label{kIndepTable}
\end{table}

\footnotetext{The bound in \cite{IndepDeg2} for the Bernoulli case is actually $\tilde O(\epsilon^{-9})$, but this can be easily improved to $O(\epsilon^{-8})$ using technology from \cite{DD}.}
Unfortunately, it is not hard to exhibit $k$-wise independent families of Bernoulli or Gaussian random variables that fail to $\epsilon$-fool the class of degree-$d$ polynomial threshold functions for $k=\Omega(d^2 \epsilon^{-2})$, putting a limit on what can be obtained through mere $k$-independence.

There have also been a number of attempts to produce pseudorandom generators by using more structure than limited independence.  In \cite{MZ}, Meka and Zuckerman develop a couple of such generators in the Bernoulli case.  Firstly, they make use of pseudorandom generators against space bounded computation to produce a generator of seed length $O(\log(n) +\log^2(\epsilon^{-1}))$ in the special case where $d=1$.  By piecing together several $k$-wise independent families, they produce a generator for arbitrary degree PTFs of seed length $2^{O(d)}\log(n) \epsilon^{-8d-3}$.  In \cite{DD}, the author develops an improved analysis of this generator allowing for a seed length as small as $O_{c,d}(\log(n)\epsilon^{-11-c})$. For the Gaussian case, the author developed a generator of seed length $2^{O_c(d)}\log(n) \epsilon^{-4-c}$ in \cite{GPRG}.  This generator was given essentially as an average several random variables each picked independently from a $k$-wise independent family of Gaussians.  The analysis of this generator was also improved in \cite{DD}, obtaining a seed length of $O_{c,d}(\log(n) \epsilon^{-2-c})$. Finally, in \cite{subpoly} it was shown that this could be improved further by taking an average with unequal weights, given seed length $O_{c,d}(\epsilon^{-c})$ for arbitrary degree and $\log(n)\exp(O(\log(1/\epsilon)^{2/3}\log\log(1/\epsilon)^{1/3}))$ for degree 2. For a summary of these results, see Table \ref{GeneralTable}.

\begin{table}[h]
\caption{Other Generators}
\begin{tabular}{|l|c|c|l|}
\hline
Paper & Bernoulli/Gaussian & d & s \\
\hline
Meka, Zuckerman \cite{MZ}  & Bernoulli & 1 & $O(\log(n) + \log^2(1/\epsilon))$ \\
Kane \cite{subpoly} & Gaussian & 1 & $O(\log(n)+\log^{3/2}(1/\epsilon))$\\
Meka, Zuckerman \cite{MZ}  & Bernoulli & d & $\log(n)2^{O(d)}\epsilon^{-8d-3}$ \\
Kane \cite{GPRG} & Gaussian & $d$ & $\log(n)2^{O(d)}\epsilon^{-4.1}$\\
Kane \cite{DD} & Gaussian & $d$ & $\log(n)O_d(\epsilon^{-2.1})$\\
Kane \cite{DD} & Bernoulli & $d$ & $\log(n)O_d(\epsilon^{-11.1})$\\
Kane \cite{subpoly} & Gaussian & 2 & $\log(n)\exp(O(\log(1/\epsilon)^{2/3}\log\log(1/\epsilon)^{1/3}))$\\
Kane \cite{subpoly} & Gaussian & $d$ & $\log(n)O_{c,d}(\epsilon^{-c})$\\
Kane, this paper & Gaussian & 2 & $O(\log^6(\epsilon)\log(n)\log\log(n/\epsilon))$\\
\hline
\end{tabular}\label{GeneralTable}
\end{table}

The bound in \cite{subpoly} came from showing that for $Y$ a weak pseudorandom generator (and in particular one that fools low degree moments) that
\begin{equation}\label{partialReplacementEqn}
\left| \E[f(X)] - \E[f(\sqrt{1-\epsilon^2}X + \epsilon Y) ]\right| \ll \epsilon^k
\end{equation}
for any $k$. This followed from an important structure theorem that said that any polynomial $p$ could be decomposed in terms of other polynomials, $q_i$ so that when the $q_i$ were localized near a random location then with high probability they would all be approximately linear polynomials. It was then shown that a moment matching random variable could fool such functions of approximately linear polynomials with high fidelity.

The bottleneck in this analysis comes in the size of the decomposition described above. On the one hand, for $d>2$ the size of the decomposition described above could potentially be quite large, though for $d=2$, it can be handled explicitly. On the other hand, the implied constant in the approximation above depends exponentially on the size of this decomposition. While, we still do not know how to solve the former problem when $d>2$, we can solve the latter in the case of degree-$2$ polynomial threshold functions.

In the special case of degree $2$ functions, we end up with a decomposition of our quadratic polynomial as a function of a single approximately linear quadratic and several other linear polynomials. Fortunately, as discovered by Meka and Zuckerman, pseudorandom generators against read once branching programs are excellent at fooling linear polynomials (or even small numbers of them). As such generators also approximately fool the expectation of low degree polynomials (which is required for dealing with the approximately linear quadratic), they will actually be much better suited as our $Y$ above. In fact, we can produce a pseudorandom generator for degree $2$ polynomial threshold functions with polylogarithmic seed length. In particular, given an appropriate notion of a discretized Gaussian (the $\delta$-approximate Gaussian defined in Section \ref{AproxGausSEc}), we have the following Theorem:

\begin{thm}\label{mainThm}
Let $\epsilon>0$ and $n$ a positive integer. For sufficiently large constant $C$, let $\delta = \log(\epsilon)/C$ and $\ell$ an integer at least $\delta^{-3}\log(\epsilon)$. For $1\leq i \leq \ell$ let $Y_i$ be a family of $n$ $\exp(-\delta^{-1}\log(n/\delta))$-approximate Gaussians seeded by a pseudorandom generator that fools read once branching programs of width $\delta^{-2}\log(n/\delta)$ to within error $\exp(-\delta^{-1}\log(n/\delta))$. Let
$$
Y = \frac{\sum_{i=1}^\ell (1-\delta^3)^{(\ell-1)/2}Y_i}{\sqrt{\sum_{i=1}^\ell (1-\delta^3)^{\ell-1}}},
$$
and let $X$ be an $n$ dimensional standard Gaussian. Then for any degree $2$ polynomial threshold function $f$ in $n$ variables,
$$
|\E[f(X)] - \E[f(Y)]| \leq \epsilon.
$$
Furthermore, such $Y$ can be constructed from generators of seed length of most $O(\log(\epsilon)^{6}\log(n)\log\log(n/\epsilon))$.
\end{thm}

In Section 2, we will go over some basic notation and results. In Section 3, we introduce the concept of an approximate Gaussian, and show that families of them seeded by a PRG for read once branching programs will fool certain functions depending on a finite numbers of linear threshold functions and polynomials of low degree. In Section 4, we will prove our generalization of Equation \eqref{partialReplacementEqn}. Finally, in Section 5, we will use this result to finish up our analysis and prove Theorem \ref{mainThm}.

\section{Background Information}

\subsection{Conventions}

Throughout the paper we will use $X,X_i,\ldots$ as standard Gaussian random variables. We will usually use $Y,Y_i,\ldots$ to denote some sort of pseudorandom Gaussian.

\subsection{Distribution of Values of Polynomials}

Given a polynomial, $p$, we will need to know some basic information about how its values at random Gaussian inputs are distributed. Perhaps the most basic measure of such distribution is the average size of $p(X)$. In order to keep track this, we will make use of the $L^t$ (and especially $L^2$) norms. In particular, recall:
\begin{defn}
If $p:\R^n\rightarrow \R$ and $t\geq 1$ then
$$
|p|_t := \left( \E[|p(X)|^t]\right)^{1/t}
$$
where $X$ is a standard Gaussian.
\end{defn}

We will also need an anticoncentration result. That is a result telling us that the value of $p(X)$ is unlikely to lie in any small neighborhood. In particular, we have:

\begin{lem}[Carbery and Wright]\label{anticoncentrationLem}
If $p$ is a degree-$d$ polynomial then
$$
\pr(|p(X)| \leq \epsilon|p|_2) = O(d\epsilon^{1/d}).
$$
Where the probability is over $X$, a standard $n$-dimensional Gaussian.
\end{lem}

We will also need a concentration result for the values. To obtain one, we make use of the hypercontractive inequality below.
The proof follows from Theorem 2 of \cite{hypercontractivity}.

\begin{lem}\label{hypercontractiveLem}
If $p$ is a degree-$d$ polynomial and $t>2$, then
$$
|p|_t \leq \sqrt{t-1}^d |p|_2.
$$
\end{lem}

This bound on higher moments allows us to prove a concentration bound on the distribution of $p(X)$.  The following result is a well-known consequence that can be found, for example, in \cite{concent}.

\begin{cor}\label{ConcentrationCor}
If $p$ is a degree-$d$ polynomial and $N>0$, then
$$
\pr_X(|p(X)| > N|p|_2) = O\left(2^{-(N/2)^{2/d}} \right).
$$
\end{cor}
\begin{proof}
Apply the Markov inequality and Lemma \ref{hypercontractiveLem} with $t = (N/2)^{2/d}$.
\end{proof}

\subsection{Hermite Polynomials}

Recall that the Hermite polynomials $h_a$ are an orthogonal set of polynomials with respect to the Gaussian distribution. Namely,
$$
\E[h_a(X)h_b(X)] = \delta_{a,b}.
$$
We will need to make use of a few standard facts about the Hermite polynomials:
\begin{itemize}
\item Any degree-$d$ polynomial, $p$, can be written as a linear combination of Hermite polynomials of degree at most $d$ so that the sum of the squares of the coefficients is $|p|_2^2$ (and thus, the sum of the absolute values of the coefficients is at most $n^d|p|_2$).
\item A Hermite polynomial of degree $d$ depends on at most $d$ coordinates of its input. In fact it can be written as a product of one variable polynomials on these inputs.
\item The sum of the absolute values of the coefficients of a Hermite polynomial of degree $d$ is $O(1)^d$.
\end{itemize}

\section{Approximate Gaussians and Read Once Branching Programs}\label{AproxGausSEc}

In order to produce a pseudorandom generator supported on a discrete set, we will first need to come up with a discrete version of the single variable Gaussian distribution. We will make use of the following notation:

\begin{defn}
We say that a random variable $Y$ is a \emph{$\delta$-approximate Gaussian}, if there is a (correlated) standard (1-dimensional) Gaussian variable $X$ so that
$$
\pr(|X-Y|>\delta) < \delta,
$$
and $|Y|=O(\log(\delta))$ with probability $1$.
\end{defn}

In particular, it is not difficult to generate a random variable with this property.

\begin{lem}\label{ApproxGaussExistLem}
There exists an explicit $\delta$-approximate Gaussian random variable that can be generated from a seed of length $O(\log(\delta))$.
\end{lem}
\begin{proof}
We assume that $\delta$ is sufficiently small since otherwise there is nothing to prove. Let $N=\lfloor \delta^{-3} \rfloor$. Note that the random variable
$$
X:= \log(z)\cos(2\pi \theta)
$$
is a random Gaussian if $z$ and $\theta$ independent uniform $(0,1)$ random variables. Let $z'$ and $\theta'$ be the roundings of $z$ and $\theta$ to the nearest half-integer multiple of $1/N$, and let
$$
Y:= \log(z')\cos(2\pi\theta').
$$
Note that $|z-z'|,|\theta-\theta'|\leq N^{-1}$. From this it follows that
$$
|X-Y| = O\left(\frac{1}{N\min(z,z')} \right).
$$
Thus, $|X-Y|<\delta$ with probability at least $1-\delta$.

On the other hand, $z'$ and $\theta'$ are discrete uniform variables with $O(\log(N))=O(\log(\delta))$ bits of entropy each. Thus, $Y$ can be generated from a seed of length $O(\log(\delta))$.
\end{proof}

We will also need to recall the concept of a read once branching program. An $(M,D,n)$-branching program is a program that is allowed to take only a single pass over an input consisting of $n$ $D$-bit blocks that is only allowed to save $M$-bits of memory between blocks. We will sometimes refer to this as a read once branching program of memory $M$ (with $n$ and $D$ usually implicit). We note that there are small seed-length generators to fool such programs. In particular, we note the following theorem of \cite{Nisan}:

\begin{thm}\label{ROBPThm}
There exists an explicit pseudorandom generator $G$ with seed length $O(M+d+\log(n/\epsilon)\log(n))$ so that if $f$ is any Boolean function computed by an $(M,D,n)$-branching program, then
$$
|\E_{X\sim_u \{\{0,1\}^D\}^n}[f(X)] - \E[f(G)]| \leq \epsilon.
$$
\end{thm}

As shown in \cite{MZ}, using pseudorandom generators for read once branching programs is a good way to fool linear threshold functions, or by extension, things that depend on a small number of linear functions of the input. They will also fool the expectations of polynomials of low degree. An important building block for our construction will be families of approximate Gaussians seeded with a pseudorandom generator which fools read once branching programs. These, it turns out will simultaneously fool functions of a small number of linear functions and expectations of low degree polynomials in the following sense:

\begin{prop}\label{indPolyProp}
Let $s$ be a quadratic polynomial in $n$ variables whose value depends on at most $r$ linear polynomials. Let $g(x)$ be the indicator function of the event that $s(x)$ lies in $I$ for some interval $I$. Let $q(x)$ be a degree $d$ polynomial in $n$ variables. Let $X$ be a standard Gaussian and let $Y$ be a family on $n$ $\delta_1$-approximate Gaussians seeded by a PRG that fools read once branching programs of length $n$ and memory $M=O((d+r)\log(n/\delta_1))$ to error at most $\delta_2$. Then
$$
\left| \E[g(X)q(X)] - \E[g(Y)q(Y)] \right| \leq O(\log(\delta_1))^{d+1}(\delta_2+n\delta_1^{1/4})n^{d}|q|_2.
$$
\end{prop}

First, we will need the following Lemma:

\begin{lem}\label{indHermlem}
Let $s$ be a quadratic polynomial in $n$ variables whose value depends on at most $r$ linear polynomials. Let $g(x)$ be the indicator function of the event that $s(x)$ lies in $I$ for some interval $I$. Let $h(x)$ be a Hermite polynomial of degree $d$. Let $X$ and $Y$ be as given in Proposition \ref{indPolyProp}. Then
$$
\left| \E[g(X)h(X)] - \E[g(Y)h(Y)] \right| \leq O(\log(\delta_1))^{d+1}(\delta_2+n\delta_1^{1/4}).
$$
\end{lem}
\begin{proof}
We prove this in two steps. First, show that for $Y'$ a family of $n$ independent approximate Gaussians that $\E[g(X)h(X)] \approx \E[g(Y')h(Y')]$. This is because by correlating $X$ and $Y'$ appropriately, we can guarantee that $X$ and $Y'$ are close with high probability. This will mean that $g(X)=g(Y')$ with high probability that that $h(X)\approx h(Y')$ with high probability. Next, we will need to show that $\E[g(Y')h(Y')]\approx \E[g(Y)h(Y)]$. This will hold because we can construct a read once branching program of small memory that computes approximations to the linear functions upon which $s$ depends and the values of the (at most $d$) coordinates upon which $h$ depends.

We may assume that $|s|_2=1$. We begin by letting $Y'$ be a family of independent $\delta_1$-approximate Gaussians. We can pick correlated copies of $X$ and $Y'$ so that with probability at least $1-n\delta_1$ each coordinate of $X$ is within $\delta_1$ of the corresponding coordinate of $Y'$. If this is the case, then $|s(X)-s(Y')| = O(n\log(\delta_1) \delta_1)$. By Lemma \ref{anticoncentrationLem}, $s(X)$ is only within this distance of an endpoint of $I$ with probability $O(n^{1/2}\delta_1^{1/2}\log^d(\delta_1))$, thus except for this probability, $g(X)=g(Y')$. Therefore, by Cauchy-Schwartz, the contribution to
$
\E[|g(X)h(X)-g(Y')h(Y')|]
$ coming from times when $g(X)\neq g(Y')$, or when some coordinate of $X$ and $Y'$ differ by more than $\delta_1$ is
$$
O((n^{1/4}\delta_1^{1/4}\log(\delta_1))\sqrt{\E[h(X)^2+h(Y')^2]}) = O(n^{1/4}\delta_1^{1/4}\log^{d+1}(\delta_1)).
$$
On the other hand
$
\E[|h(X)-h(Y')|]
$ when $X$ and $Y'$ agree to within $\delta_1$ in each coordinate is $O(n\log^d(\delta_1) \delta_1)$.
Thus,
$$
\left| \E[g(X)h(X)] - \E[g(Y')h(Y')] \right| \leq O(\log^{d+1}(\delta_1)n\delta_1^{1/4}).
$$

We now need to show that seeding $Y'$ by a read once branching programs with $M$ memory fools this expectation to within small error. Notice that a read once branching program with $O((d+r)\log(n/\delta_1))$ memory can keep track of an approximation to within $n^{-1}\delta_1^3$ of each of the $r$ normalized linear functions that $s$ depends on, and compute $h$ to precision $\delta_1$. The latter is accomplished by writing $h$ as $\prod_{i=1}^n h_{a_i}(x_i)$ and keeping track of a running product $\prod_{i=1}^m h_{a_i}(x_i)$ to relative precision $\delta_1O(\log(\delta_1))^{-d}(m/n)$. This allows the program to compute the value of $s$ to within $\delta_1$ and the value of $h$ exactly.

Thus, the probability that $h(Y')g(Y')\geq c$ is at most the probability that $h(Y)g(Y)\geq c-\delta_1$ plus the probability that $s(Y')$ is within $\delta_1$ of an endpoint of $I$ plus $\delta_2$. Note that except for an event of probability $n\delta_1$, $s(X)$ and $s(Y')$ differ by at most $O(n\log(\delta_1) \delta_1)$ and the former is this close to an endpoint of $I$ with probability at most $O(\log(\delta_1) \sqrt{n\delta_1})$. Thus, with probability $1-O(\log(\delta_1) \sqrt{n\delta_1}+n\delta_1)$, $s(Y')$ is not within $\delta_1$ of a boundary of $I$. Thus for any $c$,
$$
 \pr(h(Y)g(Y)\geq c) \leq \pr(h(Y')g(Y') \geq c - \delta_1)+ O(\delta_2+\log(\delta_1)n^{1/2}\delta_1^{1/2}+n\delta_1).
$$
Integrating this over all $|c|\leq O(\log(\delta_1))^d$ (which is the full range of values of $h(Y')$ and $h(Y)$), we find that
$$
\E[g(Y)h(Y)] \leq  \E[g(Y')h(Y')] + \delta_1 + O(\log(\delta_1))^{d+1}(\delta_2+n\delta_1^{1/2}).
$$
The lower bound follows similarly, and this completes the proof.
\end{proof}

\begin{proof}[Proof of Proposition \ref{indPolyProp}]
Note that we can write $q$ as a linear combination of degree $d$ hermite polynomials, where the sum of the absolute values of the coefficients is at most $O(n^d|q|_2)$. Our result follows from applying Lemma \ref{indHermlem} to each term separately.
\end{proof}

We also note the following corollary when $r=0$:
\begin{cor}\label{foolPolyCor}
Let $X$ and $Y$ be as in Proposition \ref{indPolyProp}. Let $q$ be a polynomial of degree at most $d$ then
$$
\left| \E[q(X)] - \E[q(Y)] \right| \leq O(\log(\delta_1))^{d+1}(\delta_2+n\delta_1^{1/4})n^{d}|q|_2.
$$
\end{cor}

\section{The Key Result}

Our analysis will depend heavily upon the following Proposition:

\begin{prop}\label{oneStepProp}
Let $\delta>0$ and $n$ a positive integer. Let $C$ be a sufficiently large constant, and let $Y$ be a family of $n$ $\exp(-C\delta^{-1}\log(n/\delta))$-approximate Gaussians seeded by a pseudorandom generator that fools read once branching programs of memory $C\delta^{-2}\log(n/\delta)$ to within error $\exp(-C\delta^{-1}\log(n/\delta))$. Let $X$ be an $n$ dimensional standard Gaussian. Then for any degree $2$ polynomial threshold function $f$ in $n$ variables, we have that
$$
\left|\E[f(X)] - \E[f(\sqrt{1-\delta^3}X+\delta^{3/2} Y)] \right| = \exp(-\Omega(\delta^{-1})).
$$
\end{prop}

We first will need to show that this result holds for a certain class of quadratic polynomials. In particular, we define:

\begin{defn}
A degree $2$ polynomial $p:\R^n\rightarrow \R$ is called \emph{$(r,\delta)$-approximately linear} if it can be written in the form
$$
p(x) = p_0(x\cdot v_1,\ldots,x\cdot v_r) + x\cdot v + q(x)
$$
for some vectors $v_1,\ldots,v_k,v$ with $v$ orthogonal to $v_i$, and some degree-$2$ polynomials $p_0$ and $q$ so that $|q|_2 < \delta |v|_2$.
\end{defn}

We now show an analogue of Proposition \ref{oneStepProp} for approximately linear polynomials:

\begin{lem}\label{approxLinLem}
Let $k,r>0$ be integers and $\delta,\delta_1, \delta_2>0$ real numbers. Let $p$ be an $(r,\sqrt\delta)$-approximately linear polynomial in $n$ variables with $f$ the corresponding threshold function. Let $X$ be an $n$-dimensional standard Gaussian, and $Y$ a family on $n$ $\delta_1$-approximate Gaussians seeded by a PRG that fools read once branching programs of length $n$ and memory $M=C(k+r)\log(n/(\delta\delta_1\delta_2))$, for sufficiently large $C$, to error at most $\delta_2$. Then
$$
\left|\E[f(X)] - \E[f(\sqrt{1-\delta^2}X+\delta Y)] \right|$$ is at most
$$
 \leq \exp(-\Omega(\delta^{-1}))4^k+ O(\log^{5k}(\delta_1)(\delta_2+n\delta_1^{1/4}))O(nk)^{4k}+O(\delta k)^{2k}+O(2^{-k/2}).
$$
\end{lem}
The basic idea of the proof is as follows. First we bin based on the approximate value of $p_0$. We are reduced to considering the expectation of the threshold function of a polynomial $C+x\cdot v + q(x)$ times the indicator function of the event that $p_0$ (a polynomial depending on a bounded number of linear functions) lies in a small interval. To deal with the threshold function, we note that averaging over possible values of $X\cdot v$ smooths it out, and we may approximate it by its Taylor polynomial. Thus, we only need $Y$ to fool the expectation of an indicator function of $p_0$ lying in a small interval, times a low degree polynomial. This should hold by Proposition \ref{indPolyProp}.
\begin{proof}
Since $p$ is $(r,\delta)$-approximately linear, after rescaling we may assume that for some orthonormal set of vectors $v,v_1,\ldots,v_k$ that
$$
p(x)=p_0(x\cdot v_1,\ldots,x\cdot v_r) + x\cdot v + q(x)
$$
for some quadratic polynomials $p_0$ and $q$ with $|q|_2<\sqrt\delta.$ We may assume that $\delta \ll 1$, for otherwise there is nothing to prove.

Let $N=2^{k}/|p|_2$. Let $I_n(x) := \textbf{1}_{p_0(x)\in [n/N,(n+1)/N)}$ and let $f_n(x) :=  I_n(x)f(x)$. Let
$$f_n^+(x) = I_n(x) \sgn(x\cdot v + q(x) + (n+1)/N),\ \textrm{  and  } \ f_n^-(x) = I_n(x) \sgn(x\cdot v + q(x) + (n)/N).$$
 Note that $f(x)=\sum_{n\in \Z} f_n(x)$. Note also that $f_n^+(x) \geq f_n(x) \geq f_n^-(x)$ for all $x,n$. We note that $f^\pm_n(x)$ is actually a very close approximation to $f_n(x)$. In particular, by Lemma \ref{anticoncentrationLem} if $X$ is a random Gaussian then
$$
\sum_{n\in \Z}\E[f_n^+(X)-f_n^-(X)] \leq \pr(|p(X)| \leq 1/N) = O(2^{-k/2} ).
$$

Thus, it suffices to show that $f_n^{\pm}(X)$ and $f_n^{\pm}(\sqrt{1-\delta^2}X+\delta Y)$ have similar expectations for each $n$. To analyze this, let $X_v$ be the component of $X$ in the $v$ direction, and $X'$ be the component in the orthogonal directions. Let
\begin{align}
g_n^{\pm}(X',Y): & = \E_{X_v}[f_n^\pm(\sqrt{1-\delta^2}X+\delta Y)]\notag \\
& = I_n(X',Y)\E_{X_v}[\sgn(C(X') + q_0(X',Y) + X_v(1+q'_1(X')+q''_1(Y))) + X_v^2q_2)]\label{gEqn}
\end{align}
where $C(X')$ is a polynomial in $X'$ and $q_0,q'_1,q''_1$ and $q_2$ are polynomials (of degree at most 2,1,1 and 0 respectively) of $L^2$ norms at most $|q_0|_2=O( \delta),|q'_1|_2=O(\sqrt\delta),|q''_1|_2=O(\delta),|q_2|_2=O(\sqrt\delta) $. We may also assume that $q_0$ is at most linear in the variables of $X'$, and that if we write $q_0(X',Y)=\delta v\cot Y+q'_0(X',Y)$, then $|q'_0(X',Y)|_2=O(\delta^{3/2}).$ We claim that with probability $1-\exp(-\Omega(\delta^{-1}))$ over the choice of $X'$ the following hold:
\begin{enumerate}
\item $\E_Y[q_0(X',Y)^2] = O(\delta^2)$.
\item $|q'_1(X')| < 1/3$.
\end{enumerate}
The first holds by Corollary \ref{ConcentrationCor} since $\E_Y[q'_0(X',Y)^2]$ is a degree $2$ polynomial in $X'$ with $L^2$ norm $O(\delta^3)$. Thus, with the desired probability $\E_Y[q'_0(X',Y)^2]=O(\delta^2)$, which implies the desired bound. The second holds by Corollary \ref{ConcentrationCor} since $q'_1$ is a degree $1$ polynomial with $L^2$ norm $O(\sqrt{\delta})$. For the next part of the argument we will assume that we have fixed a value of $X'$ so that the above holds.

Let $q_1(X',Y):=q'_1(X')+q''_1(Y)$. Note that if $|q_0(X',Y)|,|q_1(X',Y)|<2/3$, then the polynomial $C+q_0 + x(1+q_1) + x^2 q_2$ cannot have more than one root with absolute value less than $\Omega(\delta^{-1/2})$. Since $X_v$ cannot be larger than this except with probability $\exp(-\Omega(\delta^{-1}))$, the expectation above is $\erf(R)+\exp(-\Omega(\delta^{-1}))$, where $R$ is the smaller root of that quadratic. Furthermore, there will be no such root $R$ unless $|C|\ll \delta^{-1/2}$. In such a case, by the quadratic formula, this root is
\begin{equation}\label{rEqn}
R = \frac{-1-q_1+\sqrt{1+2q_1+q_1^2-4q_2(C+q_0)}}{2q_2} = (1+q_1)\frac{\sqrt{1-4q_2(C+q_0)/(1+q_1)^2}-1}{2q_2} = \frac{C+q_0}{1+q_1}+O(1).
\end{equation}
Thus, in the range $|q_0|,|q_1|<2/3$ and $|C|\ll\delta^{-1/2}$ we have that the expectation in \eqref{gEqn} is
$$
\erf(R)+\exp(-\Omega(\delta^{-1})).
$$
Note that even for complex values of $q_0$ and $q_1$ with absolute value at most $2/3$, the $\erf(R)$ (with $R$ given by Equation \eqref{rEqn}) is complex analytic with absolute value uniformly bounded. Therefore, by Taylor expanding about $q_0=0$ and $q_1=q_1'$, we can find a polynomial $P$ of degree at most $2k$ (depending on $q$, $C$ and $X'$) so that the expectation in \eqref{gEqn} is given by
\begin{align*}
& P(q_0(X',Y),q_1(X',Y)-q'_1(X')) + O(q_0(X',Y))^{2k} + O(q_1(X',Y)-q_1'(Y))^{2k}\\= & P(q_0(X',Y),q''_1(Y)) + O(q_0(X',Y))^{2k} + O(q''_1(Y))^{2k}.
\end{align*}
Furthermore, the coefficients of $P$ are all $O(1)^k$. The above must hold when $|q_0|,|q''_1|$ are not at most $1/3$. On the other hand, this means that even when $|q_0|,|q''_1|$ are larger than $1/3$, we have that $P(q_0(X',Y),q''_1(X',Y))\pm 1 = O(q_0(X',Y))^{2k} + O(q_1(X',Y))^{2k}$. This means that the above formula holds for all values of $q_0$ and $q''_1$. Thus, $g_n^{\pm}(X',Y)$ is
\begin{equation*}\label{gfullEqn}
G(Y):=\textbf{1}_{s(Y)\in I}(P(q_0(X',Y),q''_1(Y)) + O(q_0(X',Y))^{2k} + O(q''_1(Y))^{2k}) + \exp(-\Omega(\delta^{-1}))
\end{equation*}
where $s$ is some quadratic that depends on at most $r$ linear functions, $I$ is an interval. Thus, $g(X',Y)$ will be approximately the product of an indicator function of something that depends on only a limited number linear functions of $Y$ and a polynomial of bounded degree. Our proposition will hold essentially because PRGs for read once branching programs fool such functions as show in Proposition \ref{indPolyProp}.

Note that $P(q_0(Y),q''_1(Y))$ can be written as a polynomial of degree at most $4k$ and $L^2$ norm at most $O(k)^{4k}$. Letting $G_0(y)$ be
$$
G_0(y):= \textbf{1}_{s(y)\in I}P(q_0(y),q''_1(y))
$$
we have by Proposition \ref{indPolyProp} that
$$
\left|\E[G_0(X)]-\E[G_0(Y)] \right|\leq O(\log^{5k}(\delta_1)(\delta_2+n\delta_1^{1/4}))O(nk)^{4k}.
$$
Similarly, if
$$
G_1(y):= \textbf{1}_{s(y)\in I}(q_0(y)^{2k}+q''_1(y)^{2k})
$$
then
$$
\left|\E[G_0(X)]-\E[G_0(Y)] \right|\leq O(\log^{5k}(\delta_1)(\delta_2+n\delta_1^{1/4}))O(nk)^{4k}.
$$
Also,
$$
\E[G_0(X)] \leq O(\delta k)^{2k}
$$
by Lemma \ref{hypercontractiveLem}. Therefore, we have that the difference in expectations between $g_n^\pm(X',Y)$ and $g_n^\pm(X',Z)$ where $Z$ is an independent standard Gaussian, is at most
$$
\exp(-\Omega(\delta^{-1}))+O(\log^{5k}(\delta_1)(\delta_2+n\delta_1^{1/4}))O(nk)^{4k}+O(\delta k)^{2k}.
$$
Thus,
$$
\left| \E[f_n^{\pm}(X)]-\E[f_n^{\pm}(\sqrt{1-\delta^2}X+\delta Y)]\right| \leq \exp(-\Omega(\delta^{-1}))+O(\log^{5k}(\delta_1)(\delta_2+n\delta_1^{1/4}))O(nk)^{4k}+O(\delta k)^{2k}.
$$
Therefore, we have that
\begin{align*}
\sum_{|n|\leq 4^{k}} & \left| \E[f_n(X)]-\E[f_n(\sqrt{1-\delta^2}X+\delta Y)]\right| \\
&\leq \exp(-\Omega(\delta^{-1}))4^k+O(\log^{5k}(\delta_1)(\delta_2+n\delta_1^{1/4}))O(nk)^{4k}\delta^{-k}+O(\delta k)^{k}+\sum_{n}\left|\E[f_n^+(X)-f_n^-(X)]\right|\\
& \leq \exp(-\Omega(\delta^{-1}))4^k+O(\log^{5k}(\delta_1)(\delta_2+n\delta_1^{1/4}))O(nk)^{4k}\delta^{-k}+O(\delta k)^{k} +O(2^{-k/2}).
\end{align*}
On the other hand,
$$
\sum_{|n|\geq 4^k} \left| \E[f_n(X)]-\E[f_n(\sqrt{1-\delta^2}X+\delta Y)]\right|
$$
is at most the probability that either $|p_0(X)|$ or $|p_0(\sqrt{1-\delta^2}X+\delta Y)|$ is more than $2^k$ times the $L^2$ norm of $p$, which is $O(2^{-k})$ by the Markov bound and Corollary \ref{foolPolyCor}.
Thus,
\begin{align*}
\left| \E[f(X)]-\E[f(\sqrt{1-\delta^2}X+\delta Y)]\right| & \leq \sum_{|n|\in \Z}  \left| \E[f_n(X)]-\E[f_n(\sqrt{1-\delta^2}X+\delta Y)]\right|\\
 \leq \exp(-\Omega(\delta^{-1}))4^k+ & O(\log^{5k}(\delta_1)(\delta_2+n\delta_1^{1/4}))O(nk)^{4k}+O(\delta k)^{2k}+O(2^{-k/2}).
\end{align*}
As desired.

\end{proof}

We would like to reduce Proposition \ref{oneStepProp} to this case. Fortunately, it can be shown that after an appropriate random restriction that any quadratic polynomial can be made to be approximately linear with high probability.

\begin{lem}\label{approximateLinearReductionLem}
Let $p$ be a degree $2$ polynomial, $\delta>0$ and $r$ a non-negative integer. Let $X$ be a Gaussian random variable and $p^{(X)}$ be the polynomial
$$
p^{(X)}(x) := p(\sqrt{1-\delta^2}X + \delta x).
$$
Then with probability at least $1-\exp(-\Omega(r))$ over the choice of $X$, $p^{(X)}$ is $(r,O(\delta))$-approximately linear.
\end{lem}
\begin{proof}
For any polynomial $q$, let $q^{(X)}$ be the polynomial
$$
q^{(X)}(x) := q(\sqrt{1-\delta^2}X + \delta x).
$$

After diagonalizing the quadratic part of $p$ and making an orthonormal change of variables we may write
$$
p(x) = \sum_{i=1}^n p_i(x_i)
$$
where $p_i$ is a quadratic polynomial in one variable. Furthermore, we may assume that the quadratic term of $p_i(x)$ is $a_i x^2$ with $|a_i|$ decreasing in $i$. Note that
$$
p^{(X)}(x) = \sum_{i=1}^n p_i^{(X_i)}(x_i).
$$
We may write $p_i^{(X_i)}(x)$ as $\delta^2\sqrt{2}a_ih_2(x) + C_{i,1}(X_i)x + C_{i,0}(X_i)$ where $h_2(x)=(x^2-1)/\sqrt{2}$ is the second Hermite polynomial, and $C_{i,1}$ and $C_{i,0}$ are appropriate constants depending on $X_i$. Note furthermore, that unless $X_i$ lies within a small constant of the global maximum or minimum of $p_i$ that $|C_{i,1}(X_i)| = \Omega(\delta |a_i|)$. Thus, with probability at least $2/3$, independently for each $i$, we have that $|C_{i,1}(X_i)| = \Omega(\delta |a_i|)$. Let $I_i$ be the indicator random variable for the event that this happens.

From this it is easy to show that with probability $1-\exp(-\Omega(r))$ we have that $\sum_{i=1}^m I_i \geq m/2 - r$ for all $m$ (in fact the expected number of $m$ for which this fails is exponentially small). We claim that if this occurs, then $p^{(X)}$ is $(r,O(\delta))$-approximately linear. To show this, let $S$ be the set of the $r$ smallest indices $i$ for which $I_i=0$. We may write
$$
p^{(X)}(x) = \left(\sum_{i\in S} p_i^{(X_i)}(x_i)+\sum_{i\not\in S} C_{i,0}(X_i)\right) + \left(\sum_{i\not\in S} C_{i,1}(X_i) e_i \right)\cdot X + \left( \sum_{i\not\in S} \delta^2 \sqrt{2}a_i h_2(x_i)\right).
$$
We claim that letting
$$
p_0(x) = \sum_{i\in S} p_i^{(X_i)}(x_i)+\sum_{i\not\in S} C_{i,0}(X_i), \ \ \ v = \sum_{i\not\in S} C_{i,1}(X_i) e_i, \ \ \ q(x) = \sum_{i\not\in S} \delta^2 \sqrt{2}a_i h_2(x_i)
$$
shows that $p^{(X)}$ is $(r,O(\delta))$-approximately linear.

It is clear that $p_0$ depends on only the $r$ linear functions $x\cdot e_i$ for $i\in S$, that $v$ is orthogonal to these $e_i$, and that $p^{(X)}$ is the sum of $p_0, x\cdot v$ and $q$. We have only to verify that $|q|_2 = O(\delta)|v|.$ It is clear that $|q|_2 = O(\delta^2)\sqrt{\sum_{i\not\in S} a_i^2}.$ On the other hand, we have that
$$
|v|_2 = \sqrt{\sum_{i\not\in S} C_{i,1}^2(X_i)} \geq \Omega\left(\delta\sqrt{\sum_{i\not\in S}I_ia_i^2} \right).
$$
Thus, it suffices to show that
$$
\sum_{i\not\in S} I_i a_i^2 \geq \frac{1}{2} \sum_{i\not\in S} a_i^2.
$$
We can show this by Abel summation. In particular, for $i\not\in S$ let $i'$ be the value of the next smallest integer not in $S$ and let $a_{n+1}=0$. We have that
$$
\sum_{i\not\in S} a_i^2 = \sum_{i\not\in S} \sum_{j\not\in S, j\geq i} a_j^2 - a_{j'}^2 = \sum_{j\not\in S} (a_j^2-a_{j'}^2)\left(\sum_{i\not\in S,i\leq j}1\right).
$$
Similarly,
$$
\sum_{i\not\in S} I_i a_i^2 = \sum_{i\not\in S} I_i \sum_{j\not\in S, j\geq i} I_i(a_j^2 - a_{j'}^2) = \sum_{j\not\in S} (a_j^2-a_{j'}^2)\left(\sum_{i\not\in S,i\leq j}I_i\right).
$$
On the other hand, for any $j$ we have that
$$
\sum_{i\not\in S,i\leq j} I_i \geq \frac{1}{2}\sum_{i\not\in S,i\leq j} 1.
$$
Substituting into the above we find that
$$
\sum_{i\not\in S} I_i a_i^2 \geq \frac{1}{2} \sum_{i\not\in S} a_i^2
$$
and our result follows.
\end{proof}

Proposition \ref{oneStepProp} now follows easily by using Lemma \ref{approximateLinearReductionLem} to reduce us to the case handled by Lemma \ref{approxLinLem}.

\begin{proof}
Let $f(x)=\sgn(p(x))$ for some degree $2$ polynomial $p$.

Let $X_1$ and $X_2$ be independent standard Gaussians. Note that
$$
\E[f(\sqrt{1-\delta^3}X+\delta^{3/2} Y)] = \E[f(\sqrt{1-\delta}X_1 + \sqrt\delta(\sqrt{1-\delta^2}X_2+\delta Y))].
$$
Let $p^{(X_1)}$ be the polynomial given by
$$
p^{(X_1)}(x):= p(\sqrt{1-\delta}X_1+\sqrt\delta x)
$$
and let $f^{(X_1)}(x) := \sgn(p^{(X_1)})(x)$. Note that
$$
\E[f(\sqrt{1-\delta^3}X+\delta^{3/2} Y)] = \E_{X_1}[\E_{X_2,Y}[f^{(X_1)}(\sqrt{1-\delta^2}X_2+\delta Y)]].
$$
By Lemma \ref{approximateLinearReductionLem}, we have with probability $1-\exp(-\Omega(\delta^{-1}))$ over the choice of $X_1$ that $p^{(X_1)}$ is $(\delta^{-1},O(\sqrt\delta))$-approximately linear. If this is the case, then by applying Lemma \ref{approxLinLem} with $k$ a sufficiently small multiple of $\delta^{-1}$, we find that
$$
\E_{X_2,Y}[f^{(X_1)}(\sqrt{1-\delta^2}X_2+\delta Y)] = \E[f^{(X_1)}(X)] + \exp(-\Omega(\delta^{-1})).
$$
Putting these together, we find that
\begin{align*}
\E[f(\sqrt{1-\delta^3}X+\delta^{3/2} Y)] & = \E_{X_1}[\E[f^{(X_1)}(X)]] + \exp(-\Omega(\delta^{-1}))\\
& = \E[f(\sqrt{1-\delta}X_1 + \sqrt\delta X)]+ \exp(-\Omega(\delta^{-1}))\\
& = \E[f(X)] + \exp(-\Omega(\delta^{-1})).
\end{align*}
\end{proof}

\section{Cleanup}

It is not difficult to complete the analysis of our generator given Proposition \ref{oneStepProp}. We begin by applying Proposition \ref{oneStepProp} iteratively to obtain:

\begin{lem}\label{severalStepsLem}
Let $\delta>0$ and $n,\ell$ be positive integers. Let $C$ be a sufficiently large constant. For $1\leq i \leq \ell$ let $Y_i$ be an independent copy of a family of $n$ $\exp(-C\delta^{-1}\log(n/\delta))$-approximate Gaussians seeded by a pseudorandom generator that fools read once branching programs of memory $C\delta^{-2}\log(n/\delta)$ to within error $\exp(-C\delta^{-1}\log(n/\delta))$. Let $X$ be an $n$ dimensional standard Gaussian. Then for any degree $2$ polynomial threshold function $f$ in $n$ variables, we have that
$$
\left|\E[f(X)] - \E\left[f\left((1-\delta^3)^{\ell/2}X+\delta^{3/2}\sum_{i=1}^\ell (1-\delta^3)^{(\ell-1)/2}Y_i\right)\right] \right| \leq \ell\exp(-\Omega(\delta^{-1})).
$$
\end{lem}
\begin{proof}
The proof is by induction on $\ell$. The case of $\ell=0$ is trivial. Assuming that our Lemma holds for $\ell$, applying Proposition \ref{oneStepProp} to the threshold function
$$
g(x) := f\left((1-\delta^3)^{\ell/2}x+\delta^{3/2}\sum_{i=1}^\ell (1-\delta^3)^{(\ell-1)/2}Y_i\right),
$$
we find that
\begin{align*}
\E& \left[f\left((1-\delta^3)^{(\ell+1)/2}X+\delta^{3/2}\sum_{i=1}^{\ell+1} (1-\delta^4)^{(\ell-1)/2}Y_i\right)\right] \\ & \ \ \ \ \ \ \ \ \ \ =  \E\left[f\left((1-\delta^3)^{\ell/2}X+\delta^{3/2}\sum_{i=1}^\ell (1-\delta^4)^{(\ell-1)/2}Y_i\right)\right] + \exp(-\Omega(\delta^{-1}))\\
& \ \ \ \ \ \ \ \ \ \ = \E[f(X)] + (\ell+1)\exp(-\Omega(\delta^{-1})).
\end{align*}
This completes the proof.
\end{proof}

Next, we note that when $\ell$ is large, the coefficient of $X$ above is small enough that it should have negligible probability of affecting the sign of the polynomial in question.

\begin{lem}\label{FinalGeneratorLem}
Let $\delta>0$ and $n,\ell$ be positive integers. Let $C$ be a sufficiently large constant. For $1\leq i \leq \ell$ let $Y_i$ be an independent copy of a family of $n$ $\exp(-C\delta^{-1}\log(n/\delta))$-approximate Gaussians seeded by a pseudorandom generator that fools read once branching programs of memory $C\delta^{-2}\log(n/\delta)$ to within error $\exp(-C\delta^{-1}\log(n/\delta))$. Let $X$ be an $n$ dimensional standard Gaussian. Then for any degree $2$ polynomial threshold function $f$ in $n$ variables, we have that
$$
\left|\E[f(X)] - \E\left[f\left(\frac{\sum_{i=1}^\ell (1-\delta^3)^{(\ell-1)/2}Y_i}{\sqrt{\sum_{i=1}^\ell (1-\delta^3)^{\ell-1}}}\right)\right] \right| \leq \ell\exp(-\Omega(\delta^{-1})) + O((1-\delta^3)^{\ell/18}).
$$
\end{lem}
\begin{proof}
Let
$$
Y:=\frac{\sum_{i=1}^\ell (1-\delta^3)^{(\ell-1)/2}Y_i}{\sqrt{\sum_{i=1}^\ell (1-\delta^3)^{\ell-1}}},
$$
and
$$
Y' = (1-\delta^3)^{\ell/2}X + \sqrt{1-(1-\delta^3)^\ell}Y.
$$
By Lemma \ref{severalStepsLem}, it suffices to compare $\E[f(Y)]$ with $\E[f(Y')]$.
To do this, let $p$ be the degree-$2$ polynomial defining the threshold function $f$. Consider
$$
\E\left[\left(p(Y) - p(Y') \right)^2 \right].
$$
We may write this as $\E[q(X,Y_1,\ldots,Y_\ell)^2]$ for an appropriate quadratic polynomial $q$. Letting $X_1,\ldots,X_\ell$ be independent standard Gaussians, we have by repeated use of Corollary \ref{foolPolyCor} that
\begin{align*}
\E[q(X,Y_1,\ldots,Y_\ell)^2] & \leq (1+\delta^5)\E[q(X,X_1,Y_2,\ldots,Y_\ell)^2] \\
& \leq (1+\delta^5)^2\E[q(X,X_1,X_2,Y_3,\ldots,Y_\ell)^2] \\
& \leq \ldots \\
& \leq (1+\delta^5)^\ell \E[q(X,X_1,\ldots,X_\ell)^2]\\
& = (1+\delta^5)^\ell \E\left[\left(p(X)-p\left((1-\delta^4)^{\ell/2}X_1 + \sqrt{1-(1-\delta^4)^\ell}X\right)\right)^2\right]\\
& = O((1-\delta^3)^{\ell/3})|p|_2^2.
\end{align*}

Let $K=(1-\delta^3)^{\ell/9}|p|_2$. By Markov's inequality we have that $|q(X,Y_i)|\leq K$ except with probability at most $O((1-\delta^3)^{\ell/18})$. Let $f_\pm(x)=\sgn(p(x)\pm K).$ By Lemma \ref{anticoncentrationLem}, we have that $|\E[f_+(X)] - \E[f_-(X)]| \leq O(K^{1/2}) = O((1-\delta^3)^{\ell/18}).$ By Lemma \ref{severalStepsLem}, $|\E[f_\pm(X)]-\E[f_\pm(Y')]| \leq \ell\exp(-\Omega(\delta^{-1})).$ On the other hand, with high probability $|p(Y)-p(Y')|\leq K$ and thus with high probability
$$
f_+(Y') \geq f(Y) \geq f_-(Y').
$$
Therefore,
\begin{align*}
\E[f(Y)] & \leq \E[f_+(Y')] + O((1-\delta^3)^{\ell/18})\\
& \leq \E[f_+(X)] + O((1-\delta^3)^{\ell/18}) + \ell\exp(-\Omega(\delta^{-1})) \\
& \leq \E[f(X)]+ O((1-\delta^3)^{\ell/18}) + \ell\exp(-\Omega(\delta^{-1})).
\end{align*}
The lower bound follows similarly, and this completes the proof.
\end{proof}

Theorem \ref{mainThm} now follows immediately.

\begin{proof}
The result follows immediately from Lemma \ref{FinalGeneratorLem}. We can obtain the stated seed length by using the generators from Lemma \ref{ApproxGaussExistLem} and Theorem \ref{ROBPThm}.
\end{proof}

\section*{Acknowledgements}

This research was done with the support of an NSF postdoctoral fellowship.

\end{document}